\documentclass[11pt,a4paper]{article}

\usepackage{jheppub}
\usepackage{amsmath}
\usepackage{amssymb}
\usepackage{amsfonts}
\usepackage{braket}
\usepackage[normalem]{ulem}
\usepackage{color}
\usepackage{bbold}
\usepackage{ulem}
\usepackage{amsthm}
\usepackage{graphicx}

\newtheorem{theorem}{Theorem}[section]

\newcommand*{\vcenteredhbox}[1]{\begingroup
\setbox0=\hbox{#1}\parbox{\wd0}{\box0}\endgroup}

\DeclareMathOperator{\tr}{tr}

\title{Entanglement Wedge Cross Sections Require Tripartite Entanglement}

\author[1]{Chris Akers}
\author[2]{and Pratik Rath}

\affiliation[1]{Center for Theoretical Physics,\\
Massachusetts Institute of Technology, Cambridge, MA 02139, USA}
\affiliation[2]{Center for Theoretical Physics and Department of Physics,\\
University of California, Berkeley, CA 94720, U.S.A. and \\
Lawrence Berkeley National Laboratory, Berkeley, CA 94720, U.S.A.} 

\emailAdd{cakers@mit.edu}
\emailAdd{pratik\_rath@berkeley.edu}

\abstract{
We argue that holographic CFT states require a large amount of tripartite entanglement, in contrast to the conjecture that their entanglement is mostly bipartite. 
Our evidence is that this mostly-bipartite conjecture is in sharp conflict with two well-supported conjectures about the entanglement wedge cross section surface $EW$.
If $EW$ is related to either the CFT's reflected entropy or its entanglement of purification, then those quantities can differ from the mutual information at $\mathcal{O}(\frac{1}{G_N})$.
We prove that this implies holographic CFT states must have $\mathcal{O}(\frac{1}{G_N})$ amounts of tripartite entanglement.
This proof involves a new Fannes-type inequality for the reflected entropy, which itself has many interesting applications. 
}

\begin{document}
\maketitle

\section{Introduction}

We better understand quantum gravity every time we learn quantum information theoretic properties of holographic CFT states.
This is the spirit of the ``Geometry from Entanglement'' slogan \cite{VanRaamsdonk:2010pw,Maldacena:2013xja},
and it has been borne out in numerous discoveries.
At the heart of these quantum information properties is the entanglement structure of the holographic CFT state.
Know the structure explicitly, and you can in principle compute whatever quantum information property you want.

Hence it has been of great interest to probe this structure in any way tractable.
Perhaps the most famous probe is a region's von Neumann entropy, whose bulk dual is simply the area divided by $4G_N$ of the minimal-area codimension-2 surface anchored to the boundary of the region \cite{Ryu:2006ef,Ryu:2006bv}.
This is the Ryu-Takayanagi (RT) formula.
It is well-known that the RT formula places strong constraints on the entanglement structure of the CFT state \cite{Bao:2015bfa}.

That said, the von Neumann entropy is a rather coarse measure of entanglement.
It works well to quantify entanglement in a bipartite pure state, but doesn't capture all the information about entanglement structure for bipartite mixed states or multipartite states.
Hence there is much less known about the multipartite structure of entanglement in holography, owing both to the fact that there have been fewer probes of it and that it is much harder to quantify (although there has been limited progress \cite{Balasubramanian:2014hda}).


It was in this context that a particularly powerful conjecture, which we call the ``Mostly-Bipartite Conjecture'' (MBC), was made by Cui et al.~in \cite{Cui:2018dyq}.
We state this conjecture in detail now, as we understand it. 
\paragraph{Mostly-Bipartite Conjecture of \cite{Cui:2018dyq}:}
{ \it Consider a state of a holographic CFT with a gravitational dual well-described by semiclassical gravity. Let $c\sim\frac{1}{G_N}$ represent its central charge.
Given CFT subregions $A$,$B$, and $C$ with Hilbert spaces that each admit the decomposition $H_X = H_{X_1}\otimes H_{X_2} \otimes H_{X_3}$, the quantum state is ``close'' to the form
\begin{align} \label{eq:bipartite}
    \ket{\psi}_{ABC}=U_A\,U_B\,U_C\ket{\psi_1}_{A_1\,B_1} \ket{\psi_2}_{A_2\,C_1}\ket{\psi_3}_{B_2\,C_2}\ket{\widetilde{\psi}}_{A_3\,B_3\,C_3}
\end{align}
in the $G_N \to 0$ limit, where we demand that $\ket{\widetilde{\psi}}_{A_3\,B_3\,C_3}$ is `small' in the sense that its entropies are subleading in $G_N$,
\begin{align}
    S(A_3),S(B_3),S(C_3) \sim \mathcal{O}(1)~,
\end{align}
while
\begin{align}
   S(A_1)=&S(B_1)\approx \frac{I(A:B)}{2}~,\\
   S(A_2)=&S(C_1) \approx \frac{I(A:C)}{2}~,\\
   S(B_2)=&S(C_2) \approx \frac{I(B:C)}{2}~,
\end{align}
where the ``$\approx$'' symbol means at $\mathcal{O}(\frac{1}{G_N})$, and the mutual information is defined as $I(A:B) \equiv S(A) + S(B) - S(AB)$.} 

\vspace{3mm}

We will refer to this conjectured state \eqref{eq:bipartite} as the ``MBC state'' from now on.
We place quotes around ``close'' because it is not specified in what sense the states should be close. 
As we discuss in detail below, we will take this to mean close in natural distance measures usually applied to quantum states.

The motivation for this conjecture comes from the bit threads paradigm, in which Cui et al. found that an optimal bit thread configuration with the above bipartite structure exists.
Moreover, this simple entanglement structure is realized by random stabilizer tensor networks (RSTNs), which are simple toy models of holography in which the RT formula is satisfied \cite{Hayden:2016cfa,Nezami:2016zni}.

\begin{figure}[t]
\begin{center}
  \includegraphics[scale=0.5]{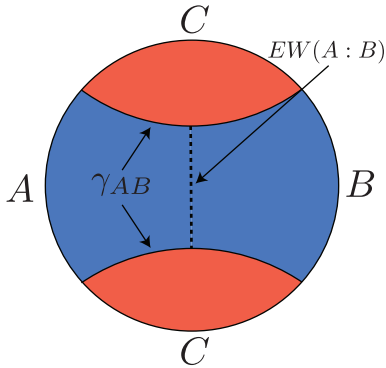} 
\end{center}
\caption{The entanglement wedge of boundary subregion $AB$ is shaded blue, while the complementary entanglement wedge, corresponding to boundary subregion $C$, is shaded red. The RT surface is $\gamma_{AB}$ (solid line), and the minimal cross section of the entanglement wedge is $EW(A:B)$ (dashed line).}
\label{fig:EW}
\end{figure}

Our goal is to argue that this entanglement structure is inconsistent with two other conjectured properties of AdS/CFT.
Both of these other conjectures relate the so-called ``minimal entanglement wedge cross section'' $EW(A:B)$, of any two CFT subregions $A$ and $B$, to information theoretic quantities of the CFT. 
We review these quantities in detail later, though see Figure~\ref{fig:EW} for a quick visual.
In the paper \cite{Dutta:2019gen}, the authors conjectured that $EW(A:B)$ equals one half a quantity called the reflected entropy, $S_R(A:B)$.
The evidence for this conjecture is very strong, and we review it later.
In the papers \cite{Takayanagi:2017knl,Nguyen:2017yqw}, the authors conjectured that $EW(A:B)$ equals a quantity called the entanglement of purification, $E_P(A:B)$.
There is also good evidence for this conjecture \cite{Bhattacharyya:2018sbw,Caputa:2018xuf,Bao:2019wcf}.
We shall refer to these as the $S_R$ and $E_P$ conjectures respectively.

Both $S_R$ and $E_P$ are more sensitive probes of multipartite entanglement than the von Neumann entropy is. 
It is this fact that places the $S_R$ and $E_P$ conjectures in tension with the MBC. 
Notably, our argument only works if either the $S_R$ or $E_P$ conjecture is true.
This is because directly computing $S_R$ and $E_P$ is difficult, so we use their respective conjectures to compute them using the bulk.

\subsection*{The Argument}

In detail, our argument proceeds in two steps.
First, we compute the reflected entropy and entanglement of purification of the state \eqref{eq:bipartite} and find that $S_R$ equals the mutual information -- and $E_P$ half the mutual information -- at leading order, $\mathcal{O}(\frac{1}{G_N})$. 
This is {\it not} true of holographic states, if either the $S_R$ or $E_P$ conjecture is correct.
It is known that $2EW(A:B) - I(A:B)$ can be non-zero at $\mathcal{O}(\frac{1}{G_N})$, which implies $S_R - I$ and $2E_P - I$ should be non-zero at leading order as well.
Therefore the MBC is in tension with the $S_R$ and $E_P$ conjectures.

That said, it is not obvious that this tension persists under small corrections to the MBC state.
Indeed, it is conceivable that some sort of small correction to \eqref{eq:bipartite} could affect its $S_R$ and $E_P$ at $\mathcal{O}(\frac{1}{G_N})$ while {\it not} affecting other quantities, such as its von Neumann entropy, at that order.
In that case, there would be no tension between these conjectures, because at any finite $G_N$ the state would be of the MBC form up to subleading corrections and also have the correct $S_R$ and $E_P$.
Something like this is true for Renyi entropies, where exponentially small changes to a state can affect the Renyi entropy at $\mathcal{O}(\frac{1}{G_N})$ but only change the von Neumann entropy an exponentially small amount.

The second step in our argument is to prove that $S_R$ and $E_P$ are not sensitive to such small changes in the state. 
More precisely, we prove that $S_R$ and $E_P$ satisfy a Fannes-like continuity inequality so that when the trace distance $\frac{1}{2}||\rho - \sigma||_1$ between $\rho$ and $\sigma$ is $\epsilon$, we have
\begin{align} \label{eq:fannes}
    |S_R(A:B)_{\rho}-S_R(A:B)_{\sigma}|\leq C_1\sqrt{\epsilon}\log d~,\\
    |E_P(A:B)_{\rho}-E_P(A:B)_{\sigma}|\leq C_2\sqrt{\epsilon}\log d~,
\end{align}
where $C_1,C_2$ are $\mathcal{O}(1)$ constants and $d$ is the dimension of $\rho$ and $\sigma$.
Moreover, we argue that $\epsilon < \mathcal{O}(1)$ if $\rho$ is a holographic CFT state and $\sigma$ is a state of the form Eqn.~\eqref{eq:bipartite}.
(Otherwise, $\rho$ would not take the MBC state form when $G_N \to 0$.)
So, even though $\log d \sim \mathcal{O}(\frac{1}{G_N})$,
the $S_R$ and $E_P$ of $\rho$ is not different from that of $\sigma$ at $\mathcal{O}(\frac{1}{G_N})$.
Therefore, small corrections to Eqn.~\eqref{eq:bipartite} that vanish as $G_N \to 0$ do not resolve the tension between these conjectures.

\subsection*{Why trace distance?}

Before proceeding, let us motivate why we use the trace distance to quantify small corrections.
The trace distance is arguably the most natural distance measure between two quantum states. 
If two states are close in trace distance, then all observables computed using one will be close to those computed using the other, inlcuding the von Neumann entropy.
Moreover, other distance measures (such as the fidelity) are quantitatively equivalent to trace distance.
There are some quantities, like the relative entropy, that quantify the similarity of two states but are not technically distance measures.
The relative entropy would work equally well for our purposes: if the relative entropy between two states is small, then their trace distance is small due to Pinsker's inequality. 

That said, there are some senses in which two states can be ``close'' without being close in trace distance.
For example, they can be ``close'' in the sense that some restricted class of observables has similar values.
It is this sense in which, for instance, ``random states'' are close to  ``Perfect states.''
Perfect states are $2n$-partite states that are maximally entangled accross any bipartition, for $n$ integer \cite{Pastawski:2015qua}.
We define a random state by acting a Haar random unitary on a fiducial $2n$-partite state.
Such random states are ``close'' to Perfect in the sense that they are nearly maximally entangled accross any bipartition.
However, they are generally far from Perfect in trace distance.\footnote{This can be seen from a simple counting argument: there are far fewer Perfect states than the total number of states. In the limit that the Hilbert space dimension goes to infinity, the average distance between any given state and the nearest Perfect state tends to zero.}

We choose not to consider ``closeness'' in this weaker sense because it is arguably against the spirit of the conjecture.
Indeed, that the von Neumann entropies of holographic CFT states match those of the MBC state was the {\it motivation} for the MBC.
The conjecture itself, as we understand it, is that the states are therefore close in some distance measure.
Inferring this stronger claim about the state from the weaker matching of entropies is what makes the conjecture so valuable.

\subsection*{Organization}
The paper is organized as follows.
We define and analyze the $S_R$ and $E_P$ conjectures in Section~\ref{sec:SR} and \ref{sec:EP} respectively. 
Also in Section~\ref{sec:SR}, we discuss why RSTNs -- which satisfy the RT formula -- fail to satisfy the $S_R$ conjecture, which naively seems like a simple application of RT.
We briefly touch on tensor networks in Section~\ref{sec:EP} as well.
Finally, we conclude with some discussion and future directions in Section~\ref{sec:discuss}.

\subsection*{Notation}

We will use the notation $S_R(A:B)$, $E_P(A:B)$ and $I(A:B)$ to denote the reflected entropy, entanglement of purification and mutual information relevant for the partition of the state about subregions $A$ and $B$.
However, in other situations where the partition is understood and we would like to make explicit the state in which these quantities are being evaluated, we shall use the notation $S_R(\rho_{AB})$, $E_P(\rho_{AB})$ and $I(\rho_{AB})$ interchangeably with the above notation.


\section{$S_R$ Conjecture vs Bipartite Entanglement} \label{sec:SR}

\subsection{Background}

We now define the reflected entropy $S_R(A:B)$.
Consider a density matrix $\rho_{AB}$ on the Hilbert space $\mathcal{H}=\mathcal{H}_A\otimes \mathcal{H}_B$.
One can define its ``canonical purification'' in a way analogous to the relationship between the thermal density matrix and the thermofield double state \cite{Dutta:2019gen}.
There exists a natural mapping between the space of linear operators acting on a $\mathcal{H}$ and the space of states on a doubled Hilbert space $\mathcal{H}\otimes \mathcal{H}'=\mathcal{H}_A\otimes \mathcal{H}_B\otimes\mathcal{H}_{A'}\otimes \mathcal{H}_{B'}$.
This mapping is sometimes labelled the channel-state duality. 
The inner product on this doubled Hilbert space is defined by
\begin{align}\label{eq:inner}
    \braket{ \rho | \sigma}_{ABA'B'}=\tr_{AB}(\rho^{\dagger}\sigma)~.
\end{align}
Thus, the operator $\sqrt{\rho_{AB}}$ can be mapped to a state $\ket{\sqrt{\rho_{AB}}}_{ABA'B'}$,
which is named the canonical purification of $\rho_{AB}$ (and is also known as the GNS state).
This state easily can be checked to reduce to the original density matrix $\rho_{AB}$ upon tracing out the subregions $A'$ and $B'$.
Given the above setup, then

\paragraph{Definition 2.1:} {\it The} reflected entropy {\it $S_R(A:B)$ is defined as
\begin{align}\label{eq:defnSR}
    S_R(A:B)=S(AA')_{\sqrt{\rho_{AB}}}=S(BB')_{\sqrt{\rho_{AB}}}~,
\end{align}
where $S(AA')_{\sqrt{\rho_{AB}}}$ is the von Neumann entropy of the reduced density matrix on the subregion $AA'$ in the state $\ket{\sqrt{\rho_{AB}}}$.}

\vspace{3mm}

In \cite{Dutta:2019gen}, it was conjectured that in AdS/CFT,
\begin{align}
    2 EW(A:B)=S_R(A:B)~,
\end{align}
where $EW(A:B)$ is the area of the ``entanglement wedge cross-section,'' i.e. the minimal-area surface that divides the entanglement wedge of $AB$ into two halves, one homologous to $A$ and the other to $B$.  
This conjecture is intuitive: the reduced density matrix of $AB$ is unchanged, and $A'B'$ has the same reduced density matrix.
One can solve the equations of motion inwards from this data local to the boundary to conclude that a viable bulk solution is the one that is simply two copies of the $AB$ entanglement wedge glued together across the extremal surface that bounds it.
(The subtleties of gluing across this extremal surface were discussed in \cite{Engelhardt:2018kcs}.) 
Applying the RT formula to the $AA'$ region of this doubled bulk implies that $S(AA')_{\sqrt{\rho_{AB}}}$ equals the area of a minimal surface dividing $AA'$ from $BB'$.
The symmetry between the entanglement wedges of $AB$ and $A'B'$ implies that this minimal surface has area $2 EW$.\footnote{Evidence for the conjecture in a time-dependent situation was provided in \cite{Kusuki:2019evw,Kusuki:2019rbk}}.

\subsection{$S_R$ of the Bipartite Entangled State}\label{sec:SRMBC}

We now compute the reflected entropy in the MBC state Eqn.~\eqref{eq:bipartite} and show that it approximately equals the mutual information,
\begin{align}
	S_R(A:B) \approx I(A:B)~.
\end{align}
This, we will argue, is incompatible with AdS/CFT.
Two properties of the reflected entropy will be useful to us.
First, it is an additive quantity under tensor products: 
\begin{align}\label{eq:additivity}
    S_R(\rho_1\otimes\rho_2)=S_R(\rho_1)+S_R(\rho_2)~.
\end{align}
This is because the canonical purification of a tensor product density matrix $\rho_1\otimes\rho_2$ is given by the tensor product state $\ket{\sqrt{\rho_1}}\otimes\ket{\sqrt{\rho_2}}$. 
Second, the reflected entropy is invariant under unitaries local to $A$ or $B$, since this is equivalent to local unitaries on $A$, $A'$, $B$ and $B'$ in the purified state.
Hence the reflected entropy of the MBC state is the same as for the state
\begin{align}
    U^\dagger_A U^\dagger_B \rho_{AB} U_A U_B =\rho_{A_1\,B_1}\otimes\rho_{A_2}\otimes\rho_{B_2}\otimes\rho_{A_3 B_3}~,
\end{align}
where e.g. $\rho_{A_2} = \tr_{C_1} \ket{\psi_2}\bra{\psi_2}_{A_2 C_1}$.
Thus, the calculation of $S_R$ splits into an individual calculation for each factor.
First consider $\rho_{A_1\,B_1}=\ket{\psi_1}\bra{\psi_1}_{A_1 B_1}$.
The canonical purification is simply a product state of two copies of $\ket{\psi_1}$, and therefore 
\begin{align}
S_{R}\left(\rho_{A_1 B_1}\right)=2S(\rho_{A_1})= I(A_1:B_1)_{\rho_{A_1 B_1}} \approx I(\rho_{AB})~.
\end{align}
Because the state $\rho_{A_2}$ only has support on $A$, its canonical purification is given by an entangled state shared between $A$ and $A'$ while $B$ and $B'$ remain trivial.
The same argument can be applied to $\rho_{B_2}$ as well.
Therefore their reflected entropies vanish, 
\begin{align}
S_{R}\left(\rho_{A_2}\right)=0
\quad \text{and} \quad
S_R(\rho_{B_2})=0~.
\end{align}
Although we have not specified any details of the state $\ket{\tilde{\psi}}_{A_3\,B_3\,C_3}$, we can use the general inequality
\begin{align}
    S_R(\rho_{A_3\,B_3})&\leq 2\,\min\{ S(\rho_{A_3}),S(\rho_{B_3})\}=O(1)
\end{align}
to put an upper bound on the contribution to $S_R$ from $\rho_{A_3 B_3}$.
It is a positive $\mathcal{O}(1)$ number, at most.
Putting everything together, we find that the reflected entropy equals 
\begin{equation}
\begin{split}
    S_R(\rho_{AB})&=S_R(\rho_{A_1\,B_1})+S_R(\rho_{A_2})+S_R(\rho_{B_2})+S_R(\rho_{A_3\,B_3})\\
		  &=I(\rho_{AB})+O(1)~.
\end{split}
\end{equation}
Hence in the $G_N \to 0$ limit, $S_R(A:B)=I(A:B)$ for the MBC state.

\begin{figure}[t]
\begin{center}
  \includegraphics[scale=1]{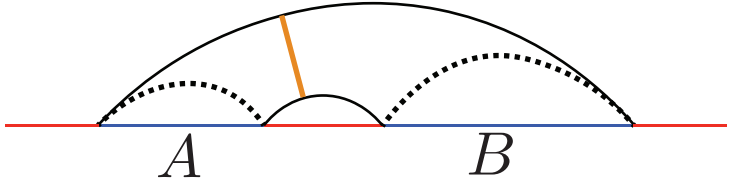} 
\end{center}
\caption{Subregion $AB$ at the threshold of a mutual information phase transition. There are two competing RT surfaces, denoted by solid and dashed black lines. The area of the dashed lines is equal to the area of the solid lines. $EW(A:B)$ before the transition is denoted by a solid orange line, while it vanishes after the transition.}
\label{fig:SRI}
\end{figure}

\subsubsection*{AdS/CFT conflict}
We now argue that this is in conflict with $S_R(A:B) = 2{EW(A:B)}$ in AdS/CFT.
The idea is that $EW(A:B)$ can be larger than $I(A:B)$ at $\mathcal{O}(\frac{1}{G_N})$. 
This is true in many generic cases, but we now provide a sharp example in which this is especially clear, from \cite{Umemoto:2019jlz}.

Consider the setup in Figure~\ref{fig:SRI}.
As one varies the distance between subregions $A$ and $B$ of a fixed size, one encounters a phase transition in the RT surfaces.
At the phase transition, both $I(A:B)$ and $EW(A:B)$ vanish.
However, at slightly shorter separations the two are quite different.
While the mutual information continuously shrinks to zero as the separation is increased, the cross-section remains $\mathcal{O}(\frac{1}{G_N})$ until exactly at the phase transition, where it discontinuously jumps to zero. 
Therefore, given $S_R(A:B) = 2 EW(A:B)$, we must conclude that the MBC state is incompatible with AdS/CFT.

\subsection{Small Corrections}\label{sec:small}

So far, we have not ruled out that the $S_R$ conjecture is consistent with the MBC state {\it with small corrections.} 
One might imagine that the reflected entropy, being non-linear in the state, could receive large corrections from terms that are subleading in $G_N$ to those in Eqn.~\eqref{eq:bipartite}.\footnote{We would like to thank Matt Headrick for discussions related to this.}
Then there would be no tension between the $S_R$ conjecture and MBC:
For any finite $G_N$, the holographic CFT state could take the form of the MBC state up to subleading terms, but its reflected entropy could be different at $\mathcal{O}(\frac{1}{G_N})$.
For comparison, this is how Renyi entropies work.
Renyi entropies are also non-linear in the state, and can change at $\mathcal{O}(\frac{1}{G_N})$ under non-perturbatively small changes to the state.

We quantify corrections to the state in terms of the natural distance measure, trace distance, defined as
\begin{align}
	T(\rho,\sigma) = \frac{1}{2}|| \rho - \sigma||_1~,
\end{align}
where $\rho$, $\sigma$ are two density matrices, and $||A||_1 = \tr(\sqrt{A^{\dagger}A})$ is the Schatten 1-norm or $L_1$ norm.
It can take values $T(\rho,\sigma) \in [0,1]$, and when the trace distance is close to $0$ then all observables are are close between the states.
If the trace distance is exactly zero, then the two states are identically equal. 
If two states admit a $G_N$ expansion, like $\rho = \rho_0 + G_N \rho_1 + \mathcal{O}(G_N^2)$, then the trace distance between them does as well:
\begin{align}
	T(\rho,\sigma) = T_0(\rho,\sigma) + G_N T_1(\rho,\sigma) + \mathcal{O}(G_N^2)~.
\end{align}
We say that two states are the same at leading order if $T_0 = 0$, i.e. $T(\rho,\sigma) \sim \mathcal{O}(G_N)$.\footnote{In fact, for the purpose of our analysis $T(\rho,\sigma) \sim \mathcal{O}(G_N^a)$ with any $a>0$ works.}
For our purposes, we could equally-well use other distance measures between states, such as the fidelity, or similarity measures like the relative entropy.

We interpret the MBC as the statement the trace distance vanishes at leading order in $G_N$ between a holographic CFT state $\rho$ and some state $\sigma$ of the form Eqn.~\eqref{eq:bipartite}.
This is for two reasons.
First, as stated above, so that $\rho$ and $\sigma$ become the same in the $G_N \to 0$ limit.
Second, because this would give a satisfactory reason for the von Neumann entropies to match at leading order (even at finite $G_N$).
(After all, this was essentially the motivation for the conjecture in the first place!)
This is due to Fannes inequality \cite{fannes1973continuity}, which states
\begin{align}
	|S(\rho)-S(\sigma)|\leq 2 T(\rho,\sigma)\log d-2 T(\rho,\sigma)\log(2 T(\rho,\sigma))~,
\end{align}
where $d$ is the dimension of $\rho$ and $\sigma$.
For holographic CFTs, $\log d \sim \mathcal{O}(\frac{1}{G_N})$, and thus if $T(\rho,\sigma) \lesssim \mathcal{O}(G_N)$, the von Neumann entropies will be guaranteed to match at $\mathcal{O}(\frac{1}{G_N})$. 

So, we are interested in whether the reflected entropy can differ at $\mathcal{O}(\frac{1}{G_N})$ between the MBC state $\sigma$ and a holographic CFT state $\rho$ that differs from it only at $\mathcal{O}(G_N)$ and higher,
\begin{align}
	T(\rho,\sigma) \sim \mathcal{O}(G_N)~.
\end{align}
We now prove this is, in fact, not possible;
the reflected entropy satisfies a continuity inequality similar to Fannes inequality for the von Neumann entropy.

\begin{theorem}[Continuity of the Reflected Entropy]\label{thm1}
Given two density matrices $\rho_{AB}$ and $\sigma_{AB}$ defined on a Hilbert space $\mathcal{H}=\mathcal{H}_A\otimes \mathcal{H}_B$ of dimension $d=d_A\,d_B$, such that $T_{AB}=T(\rho_{AB},\sigma_{AB})\leq \epsilon$, then
\[   |S_R(\rho_{AB})-S_R(\sigma_{AB})|\leq 4\sqrt{2 T_{AB}} \log(\min\{d_A,d_B\}) -2 \sqrt{2 T_{AB}} \log(T_{AB})\]
for $\epsilon\leq \frac{1}{8e^2}$.
\end{theorem}

\begin{proof}
In order to prove the above statement, we first consider the fidelity between the respective purified states $\ket{\sqrt{\rho_{AB}}}_{ABA'B'}$ and $\ket{\sqrt{\sigma_{AB}}}_{ABA'B'}$, which is given by
\begin{align}\label{eq:fidelity}
    F_{ABA'B'}&=|\braket{\sqrt{\rho_{AB}}|\sqrt{\sigma_{AB}}}|~.
\end{align}
The inner product on the canonically purified states can equivalently be computed using the original density matrices by using Eqn.~\eqref{eq:inner},
\begin{align}\label{eq:overlap}
	\braket{\sqrt{\rho_{AB}}|\sqrt{\sigma_{AB}}}&=\tr(\sqrt{\rho_{AB}}\sqrt{\sigma_{AB}})\\
    &=Q_{1/2}(\rho_{AB},\sigma_{AB}),
\end{align}
where $Q_{1/2}(\rho_{AB},\sigma_{AB})$ is defined by the above equation and is the non-commutative generalization of the Bhattacharya coefficient.\footnote{Note that $Q_{1/2}$ is a real quantity, which can be proven using cyclicity of trace and the fact that density matrices are Hermitian.}
Now we can use the inequality \cite{2012arXiv1207.1197A}
\begin{align}\label{eq:ineq}
	Q_{1/2}(\rho_{AB},\sigma_{AB})\geq&  ~1-T_{AB}\\
	\implies F_{ABA'B'}=Q_{1/2}(\rho_{AB}&,\sigma_{AB}) \geq 1-T_{AB}~.
\end{align}
This is essentially equivalent to the well known Powers-Stormer inequality.
Upon tracing out $B$ and $B'$, the fidelity monotonically increases giving us
\begin{align}\label{eq:fidtrace}
    F_{AA'}\geq F_{ABA'B'}\geq1-T_{AB}~.
\end{align}
Now, we can use another well-known inequality relating fidelity to trace distance \cite{2012arXiv1207.1197A}, giving us
\begin{align}\label{eq:fidinequality}
    T(\rho_{AA'},\sigma_{AA'})\leq \sqrt{1-F_{AA'}^2}\leq \sqrt{2 T_{AB}}~, 
\end{align}
where e.g., $\rho_{AA'}$ is the density matrix obtained by tracing out $B B'$ from the purified state $\ket{\sqrt{\rho_{AB}}}$. 
The second inequality in Eqn.~\eqref{eq:fidinequality} follows from Eqn.~\eqref{eq:fidtrace}.
Thus, starting from $\rho$ and $\sigma$ being $\epsilon$-close in trace distance on subregion $AB$, we have shown that their canonical purifications are $\sqrt{\epsilon}$-close in trace distance on subregion $AA'$.
Finally, we use Fannes inequality \cite{fannes1973continuity} to show that
\begin{equation}\label{eq:fannes1}
\begin{split}
    |S_R(A:B)_{\rho}-S_R(A:B)_{\sigma}|&=|S(\rho_{AA'})-S(\sigma_{AA'})|\\
    &\leq 2 T_{AA'} \log (d_{AA'}) - 2 T_{AA'}\log(2 T_{AA'})\\
    &\leq 4\sqrt{2 T_{AB}} \log(d_A) - 2\sqrt{2 T_{AB}}\log(T_{AB}),
\end{split}
\end{equation}
where $T_{AA'}=T(\rho_{AA'},\sigma_{AA'})$.\footnote{This result can be further tightened by using the Audenaart version of the inequality \cite{audenaert2006sharp}.}
This inequality holds for $T_{AA'}\leq\frac{1}{2e}$, which is ensured by the bound $\epsilon \leq \frac{1}{8e^2}$.
The entire analysis above was perfectly symmetric between $A$ and $B$, and from Eqn.~\eqref{eq:defnSR} we also have
\begin{align}\label{eq:fannes2}
    |S_R(A:B)_{\rho}-S_R(A:B)_{\sigma}|&\leq 4\sqrt{2 T_{AB}} \log(d_B) - 2\sqrt{2 T_{AB}}\log(T_{AB}).
\end{align}
Thus, combining Eqn.~\eqref{eq:fannes1} and Eqn.~\eqref{eq:fannes2}, we get the strengthened inequality
\begin{align}\label{eq:fannes3}
    |S_R(A:B)_{\rho}-S_R(A:B)_{\sigma}|&\leq 4\sqrt{2 T_{AB}} \log(\min\{d_A,d_B\}) - 2\sqrt{2 T_{AB}}\log(T_{AB})~,
\end{align}
which proves Theorem \ref{thm1}.
\end{proof}

Note that it was crucial that we considered the {\it canonical} purification in order for e.g. $|S(\rho_{AA'}) - S(\sigma_{AA'})|$ to have such a bound.
An arbitrary purification on $ABA'B'$ can be arbitrarily far in trace distance. 
For example, different Bell pairs purify a maximally mixed density matrix and have trace distance $1$. 
The canonical purification ensures this redundancy in basis of purification doesn't play a role here. 

We also emphasize that we have not found any examples where the inequality in Theorem~(\ref{thm1}) is saturated, despite the fact that it is easy to saturate all the individual inequalities required in proving it.
Our preliminary numerical analysis suggests that $|S_R(\rho)-S_R(\sigma)|\sim O(\epsilon)$ in all the examples that we tested, instead of the $O(\sqrt{\epsilon})$ allowed by Theorem~\ref{thm1}.
This leaves open the possibility that a tighter bound exists.
We haven't pursued a systematic numerical analysis of the above, but it would be interesting to probe this question in future.

\subsection*{Implication for AdS/CFT}

Theorem \ref{thm1} renders it impossible for two states $\rho_{AB},\sigma_{AB}$ to have reflected entropy different at $\mathcal{O}(\frac{1}{G_N})$ unless $\sqrt{T_{AB}} \log d_{AB}$ is also $\mathcal{O}(\frac{1}{G_N})$.
In a holographic CFT, $\log d_{AB} \sim \mathcal{O}(\frac{1}{G_N})$. 
So, the trace distance would need to be non-zero at leading order, $T_{AB} \sim \mathcal{O}(1)$.

However, this is not consistent with the MBC. 
Suppose $\sigma_{ABC}$ represents the density matrix corresponding to the MBC state, and $\rho_{ABC}$ represents the actual density matrix of a holographic CFT.
As we argued above, the MBC requires they should be close in the sense that $T_{ABC} \equiv T(\rho_{ABC},\sigma_{ABC}) \sim \mathcal{O}(G_N)$.
Trace distances decrease under tracing out subregions, so $T_{AB} \le T_{ABC} \sim \mathcal{O}(G_N)$.  
Therefore, $T_{AB}$ is too small for $\sigma$ and $\rho$ to have different reflected entropy at $\mathcal{O}(\frac{1}{G_N})$. 

Said differently, Theorem \ref{thm1} states that if $T_{ABC}$ is indeed $\mathcal{O}(G_N)$, then
\begin{align}\label{eq:difference}
    |S_R(\rho_{AB})-S_R(\sigma_{AB}))|&=|2 EW(A:B)-I(A:B)| \lesssim \mathcal{O}\left(\frac{1}{\sqrt{G_N}}\right)~,
\end{align}
where we have used the $S_R$ conjecture in the equality and Theorem \ref{thm1} in the inequality.
This contradicts the fact that there exist examples in AdS/CFT where $|2 EW(A:B)-I(A:B)|\sim \mathcal{O}(\frac{1}{G_N})$, e.g. the situation in Figure \ref{fig:SRI}.
Thus, we see that even small corrections to the MBC state are incapable of making it compatible with the $S_R$ conjecture.

\subsection{Tensor Networks}\label{sec:TN}

We now resolve a conundrum that our results seem to create in tensor networks.
Tensor networks have provided good toy models of holography, illustrating properties such as subregion duality and the RT formula.
In particular, a network made of perfect tensors can be shown to satisfy the RT formula under certain reasonable assumptions \cite{Pastawski:2015qua}.
Much more generally, it was shown that networks made from Haar random tensors also satisfy the RT formula \cite{Hayden:2016cfa}.

It was also emphasized in \cite{Hayden:2016cfa} that Haar randomness was overkill, and the RT formula followed simply from choosing random tensors from a 2-design ensemble, i.e. one that agrees with the first two moments of the Haar measure.
A particularly nice choice of 2-design ensemble is provided by stabilizer tensors of dimension $D=p^N$ in the limit of large $N$, where $p$ is a prime number.
Such random stabilizer tensor networks (RSTN) were further studied in \cite{Nezami:2016zni}, where it was proven that their states always take the form
\begin{align}\label{eq:RSTNstate}
\ket{\psi}_{ABC}=U_A^{\dagger} U_B^{\dagger} U_C^{\dagger} \ket{\phi^+}_{A_1\,B_1}^{\otimes n_1} \ket{\phi^+}_{A_2\,C_1}^{\otimes n_2} \ket{\phi^+}_{B_2\,C_2}^{\otimes n_3} \ket{\text{GHZ}}_{A_3\,B_3\,C_3}^{\otimes n_g}
\end{align}
where $\ket{\phi^+}$ denotes a $p$-dimensional Bell pair shared between the two parties, e.g.
\begin{align}
	\ket{\phi^+}_{A_1 B_1} \equiv \frac{1}{\sqrt{p}} \sum_{i=0}^{p-1}\ket{i}_{A_1} \ket{i}_{B_1}~,
\end{align}
and $\ket{\text{GHZ}}$ denotes a shared $p$-dimensional GHZ state, 
\begin{align}
	\ket{\text{GHZ}}_{A_3 B_3 C_3} = \frac{1}{\sqrt{p}} \sum_{i=0}^{p-1} \ket{i}_{A_3}\ket{i}_{B_3}\ket{i}_{C_3}~.
\end{align}
Neither of these states scale with $N$; they are elementary units of entanglement.
The exponents, however, can indeed have $N$-dependence.
That $N$-dependence was discovered in \cite{Nezami:2016zni}, where it was shown that in the large $N$ limit, $n_1$, $n_2$ and $n_3$ grow linearly with $N$, whereas $n_g$ remains $\mathcal{O}(1)$.
Note that $N$ here is analogous to $\frac{1}{G_N}$ in AdS/CFT.

This is exactly an MBC state like that in Eqn.~\eqref{eq:bipartite}.
Our result in Section~\ref{sec:SR} shows that this is incompatible with the conjecture $S_R = 2 EW$.
This is startling at first: the $S_R$ conjecture was motivated by the RT formula, which RSTN satisfy.
So, naively, we would expect RSTN to satisfy $S_R = 2 EW$.

\begin{figure}[t]
\begin{center}
  \includegraphics[scale=0.2]{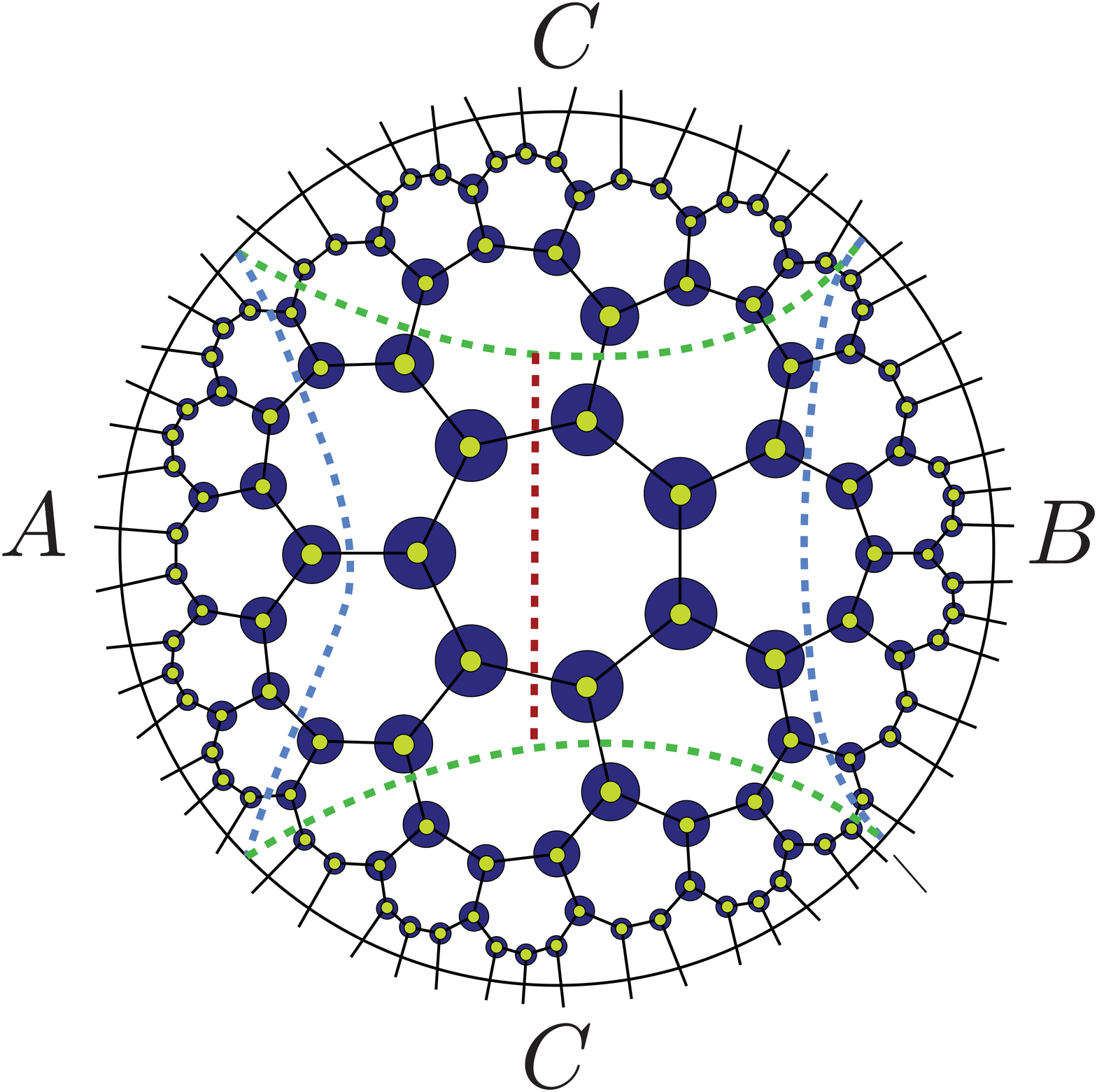} 
\end{center}
\caption{A random stabilizer tensor network with subregion $AB$ in the connected phase. The green dotted line represents the RT surface for subregion $AB$, while the yellow dotted lines represent the RT surface of $A$ and $B$ respectively. The red dotted line represents $EW(A:B)$.}
\label{fig:TNcalculation}
\end{figure}

We now compute $S_R(A:B)$ in RSTN to explain why they, in fact, do not.
The upshot will be that while the canonical purification of a state $\rho_{AB}$ is indeed given by a doubled version of its entanglement wedge (just like in AdS/CFT), the doubled entanglement wedge network does not itself satisfy RT in the naive way! 

Consider the tensor network in Figure \ref{fig:TNcalculation}.
In order to restrict to $\rho_{AB}$, we can use the fact that there is an isometry from the boundary legs of subregion $C$ to the in-plane legs cut by the RT surface of subregion $AB$.
This gives us an effective tensor network restricted to the entanglement wedge of $AB$.
In order to compute the density matrix $\rho_{AB}$, we can glue together two copies of this tensor network as in Figure \ref{fig:copies}.
The density matrix $\rho_{AB}$ has a flat entanglement spectrum as can be seen from Eqn.~\eqref{eq:RSTNstate}.
Thus, it can be shown that the operator $\sqrt{\rho_{AB}}$, and hence the canonically purified state $\ket{\sqrt{\rho_{AB}}}_{ABA'B'}$ is represented by the same doubled tensor network TN' depicted in Figure \ref{fig:copies} up to normalization.

\begin{figure}[t]
\vcenteredhbox{\includegraphics[scale=0.2]{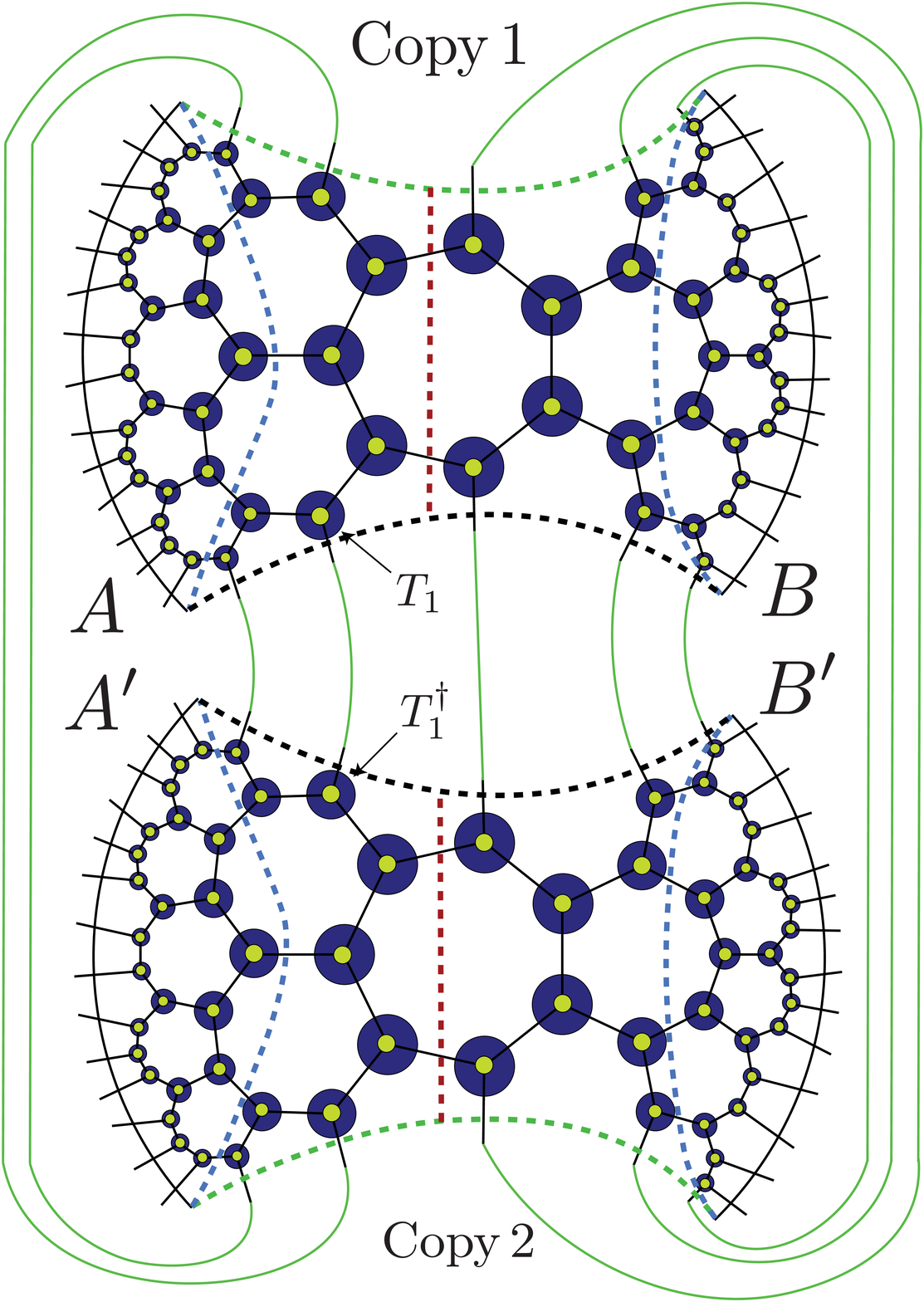} } \vcenteredhbox{\includegraphics[scale=0.2]{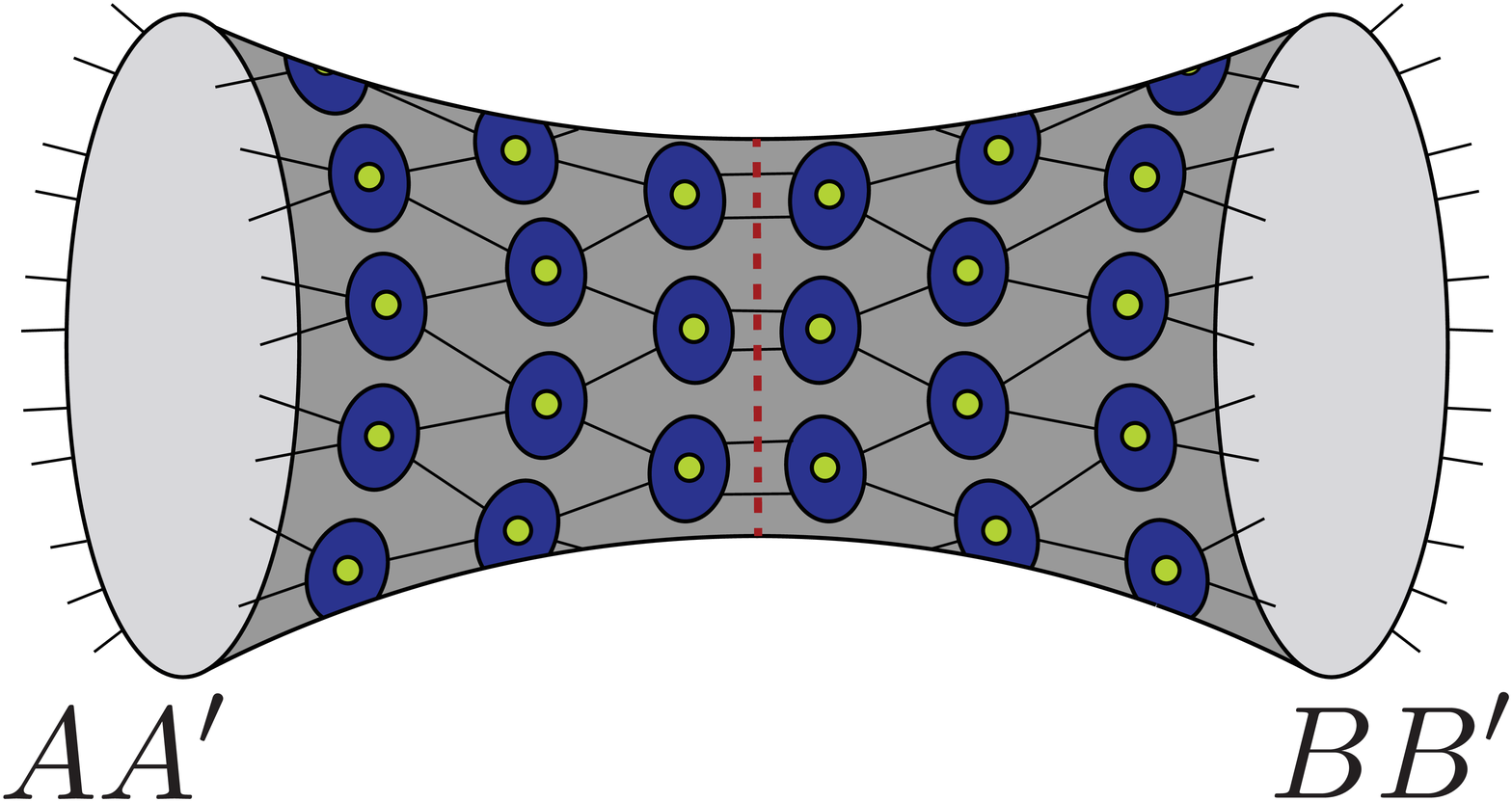}}
\caption{(Left): A reduced tensor network corresponding to the entanglement wedge of $AB$ is obtained by using the isometry from the boundary legs of subregion $C$ to the legs at the RT surface (denoted black and green dotted lines). Two copies of this RSTN glued as shown prepare the canonically purified state. We call this doubled network TN'.\newline
(Right): Geometrically, this resembles the AdS/CFT construction discussed in \protect\cite{Engelhardt:2017aux,Engelhardt:2018kcs,Dutta:2019gen}. If the RT formula holds, then $S_R(A:B)=2 EW(A:B)$.}
\label{fig:copies}
\end{figure}

TN' geometrically resembles the bulk saddle geometry obtained in the holographic construction discussed in \cite{Dutta:2019gen}.
If TN' were to satisfy the RT formula, one would indeed be led to the claim that the entropy of subregion $AA'$ is computed by the minimal cross section in this effective tensor network.
The RT surface in TN' is indeed just twice the original entanglement wedge cross section, and thus, we would have the conjectured result, $S_R(A:B)=2 EW(A:B)$.

However, this naive argument doesn't carry through because TN' has certain special properties that distinguish it from a completely random stabilizer tensor network. 
Importantly, the set of tensors used in Copy 2 in TN' are precisely correlated with the tensors in Copy 1.
E.g., in Figure \ref{fig:copies}, one can see $T_1^{\dagger}$ and $T_1$ placed at equivalent positions in either copy. 
The derivation of the RT formula depended on having completely uncorrelated tensors on both copies of the TN.

\begin{figure}[t]
\begin{center}
  \includegraphics[scale=0.3]{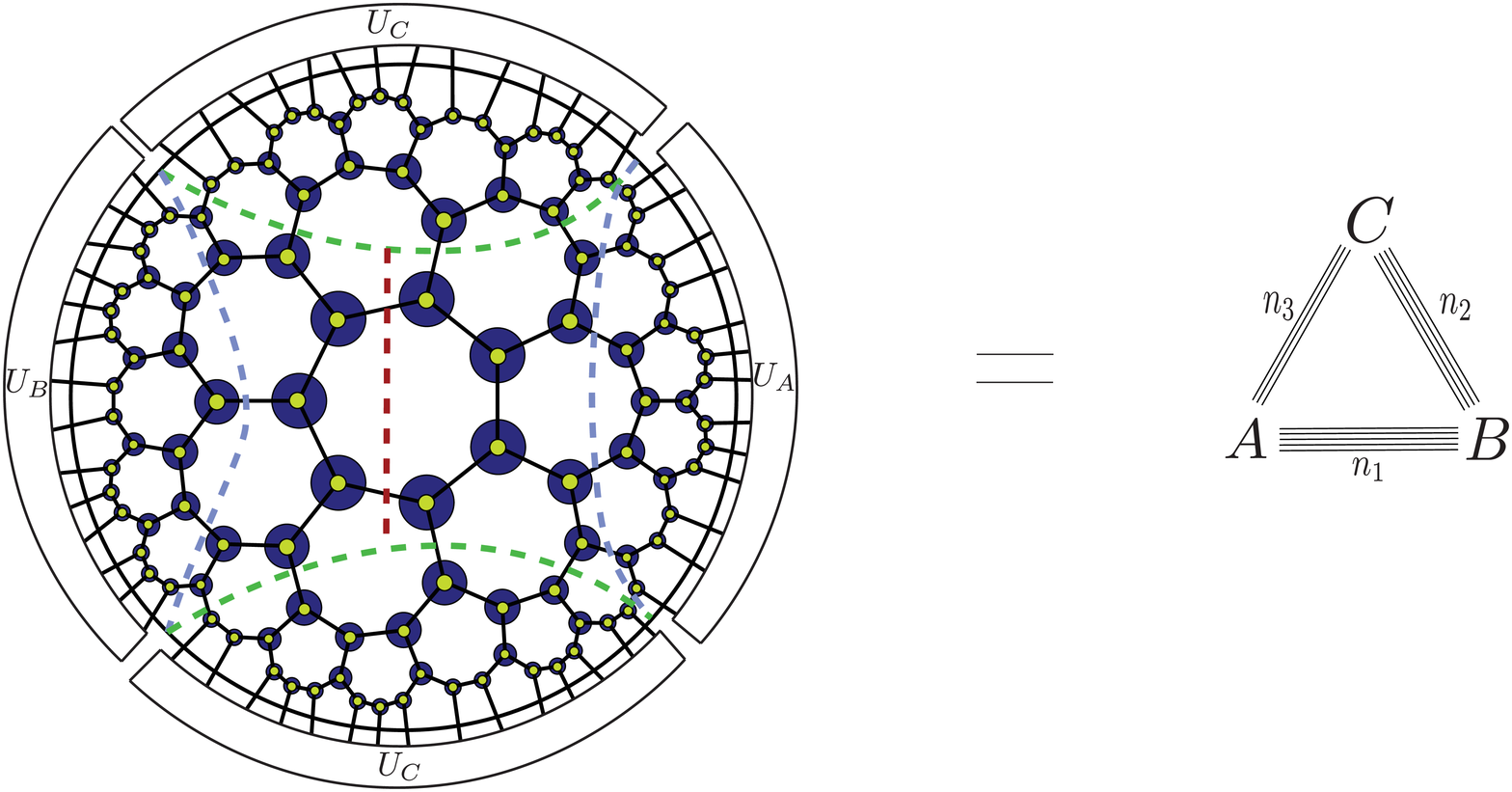}~.
\end{center}
\begin{center}
  \includegraphics[scale=0.4]{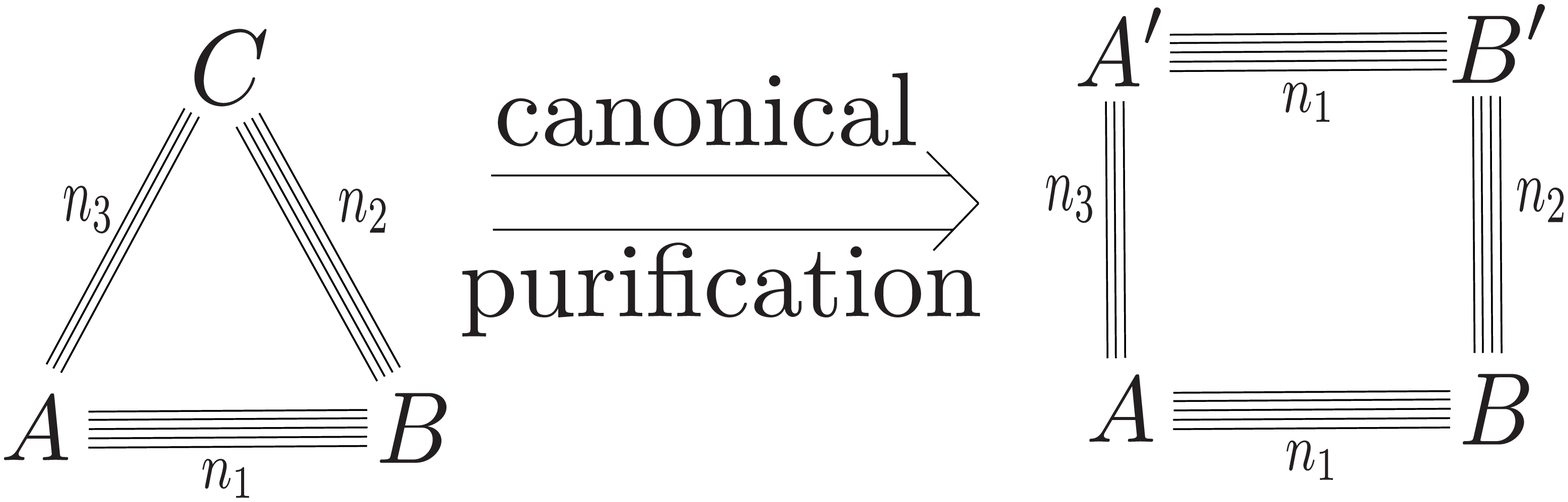} 
\end{center}
\caption{After applying local unitaries, the RSTN drastically simplifies to a combination of Bell pairs shared by the three parties. The Bell pairs then lead to a simple canonically purified state.}
\label{fig:simple}
\end{figure}

That this correlation spoils the RT formula is made manifest by the form of the state $\ket{\psi}_{ABC}$ in Eqn.~\eqref{eq:RSTNstate}. 
After applying the local unitaries, which depend sensitively on the choice
of tensors in the network, one gets a drastically simplified network
as seen in Figure \ref{fig:simple}.
The canonical purification then takes a simple form, and computing $S(AA')$ in this simple network gives us
\begin{align}\label{eq:RTSR}
    S_R(A:B)=2 n_1 \log p = I(A:B)~.
\end{align}
We see that RSTN do not satisfy $S_R = 2 EW$ because having correlated tensors precludes the application of the RT formula. 

Indeed, the RT formula in the original RSTN only required the tensors be 2-designs. 
We expect that having the tensors agree with even higher moments of the Haar measure is sufficient for the network to continue to satisfy the RT formula, even when the network is built out of many copies of itself. 
If true, then the random tensor networks of \cite{Hayden:2016cfa} should satisfy the $S_R$ conjecture, and highly random tensors -- rather than e.g. 2-designs -- would be better models of holography.
This is the subject of ongoing work \cite{new}.


\section{$E_P$ Conjecture vs Bipartite Entanglement}\label{sec:EP}
There is a tension between the $E_P$ conjecture and the MBC that is qualitatively the same as that between the $S_R$ conjecture and the MBC.
Given a density matrix $\rho_{AB}$, one can define its entanglement of purification as \cite{doi:10.1063/1.1498001} 
\begin{align}\label{eq:defnEP}
    E_P(A:B)=\min_{\ket{\psi}} S(AA')~,
\end{align}
where the minimization is over all states $\ket{\psi}_{ABA'B'}$ that are pure and consistent with the reduced density matrix $\rho_{AB}$.
In \cite{Takayanagi:2017knl,Nguyen:2017yqw}, it was conjectured that in AdS/CFT
\begin{align}\label{eq:EPconj}
    E_P(A:B)=EW(A:B).
\end{align}
This conjecture was motivated by the surface-state correspondence, wherein similar to tensor networks, a holographic state can be defined on any convex surface in the bulk \cite{Miyaji:2015yva,Nomura:2018kji,Chen:2019mis,Prudenziati:2019fev}.
Further, since the minimization over all possible purifications is a computationally intractable problem, it was assumed that minimizing over geometric purifications was sufficient (for discussions of this point, see \cite{Cheng:2019aqf}).
This conjecture, along with its multipartite generalizations, has received a lot of attention recently, although proofs or related computations have generally required various strong assumptions \cite{Guo:2019pfl,BabaeiVelni:2019pkw,Bao:2018fso,Bao:2018gck,Bao:2019wcf,Umemoto:2018jpc,Bhattacharyya:2018sbw,Caputa:2018xuf,Bhattacharyya:2019tsi,Harper:2019lff,Hirai:2018jwy}.

To argue that the $E_P$ conjecture is incompatible with the MBC, we review results that are essentially known in the literature.
This distinguishes this argument from the one in Section~\ref{sec:SR}, which involved our Theorem \ref{thm1} that was completely new.

In order to compute $E_P(A:B)$ in the MBC state, we first note that $E_P$ is a sub-additive quantity under tensor products \cite{bagchi2015monogamy}.
In fact, additivity holds for pure states, $\rho_{AB}=\ket{\psi}_{AB}\bra{\psi}_{AB}$, and completely decoupled states, $\rho_{AB}=\rho_{A}\otimes\rho_{B}$, but not in general \cite{chen2012non}.
Using this property, we find for the MBC state
\begin{align}\label{eq:subadditive}
    E_P(\rho_{AB})&\leq E_P(\rho_{A_1\,B_1}) + E_P(\rho_{A_2})+ E_P(\rho_{B_2}) +E_P(\rho_{A_3\,B_3}).
\end{align}
The first term on the right hand side gives $E_P(\rho_{A_1 B_1})=S(\rho_{A_1})=\frac{1}{2}I(A_1:B_1)$, because $\rho_{A_1 B_1}$ is a pure state. 
The second and third terms involve only one of either $A$ or $B$ and thus give $E_P(\rho_{A_2})=E_P(\rho_{B_2})=0$. 
The fourth term can be bounded using the known inequalities for $E_P$ to obtain
\begin{align}\label{eq:EPbounds}
    0\leq& E_P(\rho_{A_3\,B_3})\leq 2\min\{S(\rho_{A_3}),S(\rho_{B_3})\},
\end{align}
and thus, $E_P(\rho_{A_3\,B_3})$ is an $O(1)$ positive quantity.
Putting these results together and using known inequalities, we find that
\begin{align}
    \frac{1}{2}I(A:B) &\leq E_P(\rho_{AB})\leq \frac{1}{2}I(A:B)+O(1).
\end{align}
Thus, for $G_N \to 0$, we obtain $E_P(A:B) \approx \frac{1}{2}I(A:B)$, where $``\approx''$ denotes matching at $\mathcal{O}(\frac{1}{G_N})$. 
Similar to the result in Section~\ref{sec:SRMBC}, we find that the MBC state is incompatible with the $E_P$ conjecture.

\subsection*{Small Corrections}

One might again worry that small corrections to the MBC state might make it compatible with the $E_P$ conjecture.
However, this too can be ruled out by the following theorem.

\begin{theorem}[Continuity of the Entanglement of Purification]\label{thm2}
Given two density matrices $\rho_{AB}$ and $\sigma_{AB}$ defined on a Hilbert space $\mathcal{H}=\mathcal{H_A}\otimes \mathcal{H_B}$ of dimension $d=d_A\,d_B$, such that $T_{AB}=T(\rho_{AB},\sigma_{AB})\leq \epsilon$, then
\[   |E_P(\rho_{AB})-E_P(\sigma_{AB})|\leq 40\sqrt{T_{AB}} \log(d) -4 \sqrt{T_{AB}} \log(4\sqrt{T_{AB}})\]
for $\epsilon\leq \frac{1}{4e^2}$.
\end{theorem}

\begin{proof}
This proof essentially follows from Theorem 1 of \cite{doi:10.1063/1.1498001}, where it was shown that
\begin{align}
    |E_P(\rho_{AB})-E_P(\sigma_{AB})|&\leq 20 D(\rho_{AB},\sigma_{AB})\log(d) - D(\rho_{AB},\sigma_{AB})\log(D(\rho_{AB},\sigma_{AB})) 
\end{align}
where $D(\rho_{AB},\sigma_{AB})=2\sqrt{1-F_{AB}}$ is the Bures distance.
Using the inequality
\begin{align}
    1-T_{AB}&\leq F_{AB}\implies D(\rho_{AB},\sigma_{AB})\leq 2\sqrt{T_{AB}},
\end{align}
we obtain the desired result.
\end{proof}

Using Theorem \ref{thm2}, we conclude that a slightly-corrected MBC state is still incompatible with the $E_P$ conjecture, by a similar argument to the one made in Section~\ref{sec:small}.

\subsection*{Tensor Networks}
$E_P$ is a difficult quantity to compute in general, and hence it is much harder to understand the tensor network story analogous to that in Section~\ref{sec:TN}.
However, in the case of RSTNs, the simplified network (obtained by applying local unitaries, as in Figure \ref{fig:simple}) has an $E_P$ that can easily be calculated to give $\frac{1}{2}I(A:B)$ at leading order in $G_N$.

It is important to note that the $E_P$ conjecture was originally motivated by restricting to geometric purifications and computing the optimal RT surface anchored to the entanglement wedge.
An important insight we gain here is that non-geometric tensor networks like the simplified network were crucial for the minimization in computing $E_P$, at least for RSTNs.
It would be interesting to understand if this is more generally true \cite{Cheng:2019aqf}.

\section{Discussion}\label{sec:discuss}

We have provided two pieces of evidence that suggest that holographic states require a large amount of tripartite entanglement:
Having little tripartite entanglement is inconsistent with both the strongly-supported conjectures that $S_R = 2 EW$ and $E_P = EW$. 
We now focus on some of the caveats, implications, and interesting future directions stemming from this work.

\subsection*{Trace distance}

We have demonstrated that holographic CFT states cannot be close in trace distance to the MBC state.
It is still possible that they are ``close'' in another sense.
Being close in trace distance is a strong criterion that ensures closeness in all observable quantities and is a standard measure of similarity of states in quantum information.
If we allow weaker conditions of closeness on the state, such as closeness in a restricted class of observable quantities, it might be possible to make the MBC state consistent with the $S_R$ and $E_P$ conjectures.
However, we do not see any evidence for other quantities that may be reproduced by assuming an MBC state, and in particular, measures of multipartite entanglement are in conflict with the conjectured state.
It would be interesting to see if other weaker forms of closeness can lead to a version of the MBC that is both useful and compatible with the other two conjectures.

\subsection*{Limitation on Tensor Networks}

This analysis also illuminates limitations of tensor networks as toy models of holography.
Since the von Neumann entropy is a reasonably coarse grained quantity, even 2-design tensor networks such as random stabilizer tensor networks were able to reproduce the RT formula.
However, stabilizers are a very special class of tensors, and are generically far in trace distance from Haar random tensors (owing to the fact that there are many more Haar random tensors than stabilizers).
Hence, properties from any such tensor networks should be considered carefully, because they may not agree with actual holographic answers.

In fact, specific tensor network models have previously been used to model ``mostly bipartite" entanglement that arises in certain regions of moduli space of multiboundary wormholes \cite{Marolf:2015vma,Peach:2017npp}.
It would be interesting to explore whether more refined tensor network models can capture the right form of multipartite entanglement employed by holographic states.

It is interesting to note that the tensor network in \cite{Bao:2019fpq,Bao:2018pvs} is close in trace distance to the holographic state, by construction.
Certain classes of their tensor networks require the $E_P=EW$ conjecture, so it would be interesting to repeat the above analysis in their case.

\subsection*{Entanglement measures}

As we saw in our analysis in Section~\ref{sec:SR}, the reflected entropy $S_R(A:B)$ is a much more fine-grained entanglement measure than individual entanglement entropies, for mixed density matrices.
This quantity is very naturally motivated from holography and hasn't yet been studied in the quantum information literature.
In this sense, it is similar to the refined Renyi entropies which is also a very natural quantity in holography, but hasn't been analyzed in quantum information \cite{Dong:2016fnf,Akers:2018fow,Dong:2018seb,Bao:2019aol}.
It would be interesting to understand its properties and generic behaviour in quantum systems.

There is, in fact, a zoo of quantities that measure multipartite entanglement and there is not a clear understanding of a canonically best choice.
Owing to this fact, there have been many proposals in holography for such quantities including, among many others, the entanglement negativity and odd entropy \cite{Tamaoka:2018ned,Kusuki:2019zsp,Levin:2019krg,Umemoto:2019jlz}.
Similarly, higher party versions of the reflected entropy have also been proposed, motivated by AdS/CFT \cite{Marolf:2019zoo,Bao:2019zqc, Chu:2019etd}.
It would be interesting to understand each of these quantities in the context of holography, or even toy models such as tensor networks.
If the program of understanding quantum gravity by understanding quantum information is to progress, it is crucial that we obtain a more refined understanding of multipartite entanglement measures.

\subsection*{Applications for reflected entropy continuity}

Our new bound in Theorem \ref{thm1} has many interesting applications.
For example, it might be useful in proving inequalities about $S_R$ that were conjectured in \cite{Dutta:2019gen}.
Indeed, those inequalities might be easier to prove for e.g. the fixed-area states defined in \cite{Akers:2018fow,Dong:2018seb}. 
Holographic CFT states are generally close in trace distance to one fixed-area state.
So, bounds on the reflected entropy of one translate to bounds on the reflected entropy of the other.
It would be interesting to find other uses for this theorem. 

\subsection*{GHZ isn't enough}

While we have demonstrated that tripartite entanglement is necessary for the $S_R$ and $E_P$ conjectures, we have not emphasized what type of tripartite entanglement is required.
In fact, GHZ entanglement -- even a lot of it -- does not help. 
One can show that GHZ entanglement also satisfies $S_R(A:B)=I(A:B)$.
(Note that this problem is also not resolved by adding superselection sectors, similar to the $\alpha$ blocks in operator-algebra quantum error correction \cite{Harlow:2016vwg,Akers:2018fow,Dong:2018seb,Dong:2019piw}.
These results strongly suggest that the ``stabilizerness" of holographic states is very low, which will be discussed in upcoming work \cite{brian}.\footnote{We thank Brian Swingle for discussions related to this.})

Beyond this, there is little we can say.
It is difficult to pinpoint what type of entanglement must be present, because there are many inequivalent forms of tripartite entanglement, and the classification is not well understood in general.
In the case of three qubits, there are just two inequivalent forms of entanglement: GHZ and W \cite{dur2000three}.
For $A,B$ two of the three qubits in a W-state, $S_R(A:B)= 1.49 \log 2$ while $I(A:B) = 0.92 \log 2$, and therefore W-entanglement might be used to alleviate the gap between the MBC (and RSTNs) and holography.
Similarly, numerical analyses suggest that $E_P(A:B)\neq \frac{1}{2}I(A:B)$ for such states \cite{doi:10.1063/1.1498001,chen2012non}.
It would be interesting understand better the particular kind of tripartite entanglement that is crucial for holography.

\section*{Acknowledgements}
We thank Ning Bao, Raphael Bousso, Ven Chandrasekaran, Netta Engelhardt, Tom Faulkner, Matthew Headrick, Arvin Shahbazi Moghaddam, Yasunori Nomura and Brian Swingle for helpful discussions and comments.
C.A. is supported by the US Department of Energy grants DE-SC0018944 and DE-SC0019127, and also
the Simons foundation as a member of the It from Qubit collaboration. 
This work was supported in part by the Berkeley Center for Theoretical Physics; by the Department of Energy, Office of Science, Office of High Energy Physics under QuantISED Award DE-SC0019380 and under contract DE-AC02-05CH11231; and by the National Science Foundation under grant PHY1820912.

\bibliographystyle{JHEP}
\bibliography{mybibliography}

\end{document}